\documentclass{LMCS}

\def\dOi{9(3:22)2013}
\lmcsheading%
{\dOi}
{1--18}
{}
{}
{Apr.~\phantom02, 2012}
{Sep.~18, 2013}
{}

\ACMCCS{[{\bf Theory of computation}]: Logic---Equational logic and
  rewriting; Semantics and reasoning---Program reasoning; Semantics
  and reasoning---Program semantics---Categorical semantics}

\subjclass{
  F.3.1 [Logics and Meanings of Programs]: Specifying and Verifying and 
  Reasoning about Programs;
F.3.2 [Logics and Meanings of Programs]: Semantics of Programming Languages
    ---   algebraic approaches to semantics, denotational semantics; 
General Terms: Theory}

\usepackage{amsmath,amssymb}
\usepackage{stmaryrd,latexsym,marvosym}		
\newcommand{\hash}{\#}
\usepackage{wasysym}
\usepackage{xspace}
\usepackage[english]{babel}
\usepackage{url}
\usepackage{microtype}
\usepackage{slashed}   
\usepackage{stackrel}  
\usepackage{alltt}
\usepackage{ifdraft}   
\usepackage[all]{xy}
\usepackage{proof}
\usepackage{enumerate,hyperref}

\ifdraft{
\usepackage[show]{ed}
}{\usepackage[hide]{ed}}

\theoremstyle{plain}
\newtheorem{lemdefn}[thm]{Lemma and Definition}

\theoremstyle{definition}

\newtheorem{remdefn}[thm]{Remark and Definition}

\newcommand{\CC}{{\Cls C}}
\newcommand{\CE}{{\Cls E}}

\newcommand{\CM}{{\Cls M}}

\newcommand{\CT}{{\Cls T}}

\newcommand{\CW}{{\Cls W}}
\newcommand{\CTone}{\Cls T}
\newcommand{\CTtwo}{\Cls S}
\newcommand{\Theorify}[1]{\Psi{#1}}

\newcommand{\Monadone}{T}
\newcommand{\Monadtwo}{S}
\newcommand{\Amonad}{T}

\newcommand{\Set}{\Cat{Set}}
\newcommand{\Class}{\Cat{Class}}

\newcommand{\TAlg}{\Cat{\text{-}TAlg}}
\newcommand{\CatMod}[2]{{#2}\Cat{\text{-}Alg}\,[{#1}]}
\newcommand{\TAlgOf}[3]{({#2},{#3})\Cat{\text{-}TAlg}\,[{#1}]}
\newcommand{\SmallTalgOf}[2]{({#1},{#2})\Cat{\text{-}TAlg}}
\newcommand{\congsp}[1]{\quad \cong_{#1} \quad }
\newcommand{\eqdefsp}{\quad \stackrel{\mbox{\rm {\tiny def}}}{=} \quad }

\newcommand{\Cat}{\mathbf}
\newcommand{\Cls}{\mathcal}

\newcommand{\id}{\operatorname{id}}
\newcommand{\PSet}{{\mathcal P}}

\newcommand{\consto}{\mathrel{\ooalign{$\circ$\cr\kern-1.65pt$\rightarrow$}}}

\newcommand{\tensor}{\varotimes}

\newcommand{\tpair}[2]{{#1\wideparen{\;}\kern-1pt #2}}
\newcommand{\argument}{\_\!\_}

\newcommand{\downset}{\kern-2.8pt\downarrow\kern-2.8pt{}} 
\newcommand{\upset}{\kern-2.8pt\uparrow\kern-2.8pt{}} 

\newcommand{\unit}{\star}
\newcommand{\fst}{\operatorname{\sf fst}}
\newcommand{\snd}{\operatorname{\sf snd}}

\newcommand{\CASE}{\operatorname{\sf case}}
\newcommand{\OF}{\operatorname{\sf of}}
\newcommand{\case}[3]{\CASE\kern1.2pt #1\kern1.2pt \OF\kern1.2pt #2;\kern1.2pt #3}

\newcommand{\caseOne}[2]{\CASE\kern1.2pt #1\kern1.2pt \OF\kern1.2pt #2}

\newcommand{\DO}{\operatorname{\sf do}}
\newcommand{\retOp}{\operatorname{\sf ret}}
\newcommand{\ret}[1]{\retOp #1}
\newcommand{\letTerm}[2]{\DO\kern1.2pt#1; #2}
\newcommand{\letTermR}[2]{#1;\,#2}
\newcommand{\leteq}{\gets}

\newcommand{\bind}[2]{#1\leteq #2}

\newcommand{\letTermO}[2]{\DO_{\nu}\kern1.2pt#1; #2}
\newcommand{\LOOP}{\operatorname{\sf loop}}
\newcommand{\THEN}{\operatorname{\sf then}}

\makeatletter
\long\def\loopTerm@[#1][#2][#3]#4#5#6{
\DO\kern1.2pt#4 #1 \LOOP #3 #5 #2\THEN #6
}

\newcommand{\loopTerm}{
\optparams{\loopTerm@}{[;][\};][\{]}
}
\makeatother

\newcommand{\LET}{\operatorname{\sf init}}
\newcommand{\letLoop}[2]{\LET #1\kern1.2pt\LOOP\kern1.2pt\{#2\}}
\newcommand{\LetLoop}[2]{\LET #1\kern1.2pt\LOOP\kern1.2pt\bigl\{#2\bigr\}}

\newcommand{\IF}{\operatorname{\sf if}}
\newcommand{\ifTerm}[3]{\IF\kern1.2pt #1\kern4.2pt {\sf then}\kern2.2pt #2\kern4.2pt {\sf else}\kern2.2pt #3}
\newcommand{\ifTermO}[3]{\IF_{\nu} #1\kern2.2pt {\sf then}\kern2.2pt #2\kern2.2pt {\sf else}\kern2.2pt #3}

\newcommand{\WHILE}{\operatorname{\sf while}_{\nu}}
\newcommand{\whileTerm}[3]{\LET #1\kern1.2pt\WHILE\kern1.2pt #2 \kern2.2pt{\sf do}\kern2.2pt #3}
\newcommand{\whileTermS}[2]{\WHILE\kern1.2pt #1 \kern2.2pt{\sf do}\kern2.2pt #2}

\newcommand{\SEQ}{\operatorname{\sf seq}}
\newcommand{\seqTerm}[2]{\SEQ_{\nu}\kern1.2pt#1; #2}

\newcommand{\infrule}[2]{\frac{#1}{#2}}

\newcommand{\brks}[1]{\langle #1\rangle}

\newcommand{\Path}[1]{\Bigl(\kern-2.1pt\Bigl|\, #1\,\Bigr|\kern-2.1pt\Bigr)}

\newlength{\myboxwidth}
{\end{center}\end{figure}}

\newcommand{\inject}{\hookrightarrow}
\newcommand{\ulist}{\iota}

\newcommand{\len}[1]{\hash{#1}}

\DeclareMathOperator{\redpar}{{
\declareslashed{}{
\vrule height2pt depth 2pt
\kern1pt
\vrule height2pt depth 2pt
}{0}{0}{\rightarrowtail}\slashed{\rightarrowtail}}}

\renewcommand{\emptyset}{\varnothing} 

\newcommand{\Nat}{\mathbb{N}}

\newcommand{\op}{{\mathit{op}}}

        {\begin{enumerate}}%
        {\end{enumerate}}
        {\begin{enumerate}}%
        {\end{enumerate}}

\newcommand{\into}{\hookrightarrow}

\newcommand{\cofinal}{\mathit{cf}}

\usepackage{ifpdf}					
\ifx\todo\undefined
\def\todo{}
\fi
\ifpdf
  \usepackage{color} 
  \renewcommand{\todo}[1]{
  {\color{red}{[TODO: #1]}}
  }
\else
  \usepackage{color}
  \renewcommand{\todo}[1]{
  {\color{red}{[TODO: #1]}}
  }
\fi

\newcommand{\varstar}{\star}

\newcommand{\pow}{\PSet^{\varstar}_{\kappa}}
\newcommand{\Nbb}{\mathbb N_{>0}}

\newcommand{\itum}{\item}
\newcommand{\easy}{}

\begin{document}
\sloppy


\title[Exploring the Boundaries of Monad Tensorability on Set]{Exploring the Boundaries of\\ Monad Tensorability on Set}

\author[N.~Bowler]{Nathan Bowler\rsuper a}
\address{{\lsuper a}Department of Mathematics, Universit\"at Hamburg}
\email{N.Bowler1729@gmail.com}

\author[S.~Goncharov]{Sergey Goncharov\rsuper b}
\address{{\lsuper{b,d}}Department of Computer Science, Friedrich-Alexander-Universit\"at Erlangen-N\"urnberg}
\email{\{Sergey.Goncharov, Lutz.Schroeder\}@fau.de}

\author[P.~B.~Levy]{Paul Blain Levy\rsuper c}
\address{{\lsuper c}School of Computer Science, University of Birmingham}
\email{P.B.Levy@cs.bham.ac.uk}
\thanks{{\lsuper c}Paul Blain Levy supported by EPSRC Advanced Research Fellowship EP/E056091}

\author[L.~Schr\"oder]{Lutz Schr\"oder\rsuper d}
\address{\vspace{-18 pt}}

\keywords{Monads, tensor products, side effects, non-determinism}

\begin{abstract}
  We study a composition operation on monads,
  equivalently presented as large equational theories. Specifically, we discuss the
  existence of \emph{tensors}, which are combinations of theories that
  impose mutual commutation of the operations from the component
  theories. As such, they extend the \emph{sum} of two 
  theories, which is just their unrestrained combination. Tensors of
   theories arise in several contexts; in particular, in the
  semantics of programming languages, the monad transformer for global
  state is given by a tensor. We present two main results: we show
  that the tensor of two monads need not in general exist by
  presenting two counterexamples, one of them involving finite
  powerset (i.e.\ the theory of join semilattices); this solves a
  somewhat long-standing open problem, and contrasts with recent
  results that had ruled out previously expected counterexamples.  On
  the other hand, we show that tensors with bounded powerset monads do
  exist from countable powerset upwards.
\end{abstract}

\maketitle

\section{Introduction}
\noindent The concept of monad may be regarded as a category-theoretic abstraction of the notion of equational theory that is insensitive to the choice of syntax. More recently, monads have gained importance in the theory and practice of programming, where they are now commonly recognized as a standard formal abstraction for computational effects. In this context, the combination of theories, or monads, can be seen as modelling the combination of computational effects, a topic of interest e.g.\ in the modular semantics of programming languages. 

One way to implement the combination of effects is via \emph{monad
  transformers}~\cite{CenciarelliMoggi93,Moggi95}. Essentially, a
monad transformer is a function that sends monads to monads
(additional properties, such as functoriality, are not in general
imposed). Monad transformers form the core of the treatment of side-effects in the
functional programming language Haskell~\cite{Peyton-Jones03}.
Besides the fact that due to their lack of structure they say only very little about the mathematical foundations of monad combination, monad transformers have been criticized for their asymmetry~\cite{HylandPlotkinEtAl06} which treats one set of effects as the monad transformer and the other set of effects as an argument of the latter. E.g.\ the approach via monad transformers hides the symmetry in the combination of exceptions and I/O: this combination can be obtained either by applying the I/O monad transformer to the exception monad or by applying the exception monad transformer to the I/O monad, but this equivalence is not apparent from the corresponding monad transformer expressions. Moreover, monad transformers are ad hoc in character, and have been known only for a limited number of effects, a prominent negative example being nondeterminism. 

It has turned out, however, that some of the most important monad
transformers have an elegant abstract description using \emph{sum} and
\emph{tensor}.  Specifically, the monad transformers for exceptions
and I/O constructs a sum, and the monad transformers for state,
reader, and writer construct a
tensor~\cite{HylandPlotkinEtAl06,HylandLevyEtAl07}. Whereas the sum of
two monads is the simplest monad supporting both given effects without
any interaction between them (and corresponds, in terms of 
theories, just to taking the disjoint union), the tensor (whose
definition goes back to~\cite{Freyd66}) moreover requires commutation
of these effects over each other, e.g.\ in case of tensoring
statefulness with finite nondeterminism one has
\begin{displaymath}
(\mathtt{x := a; (b + c)}) ~=~ (\mathtt{(x := a; b) + (x := a; c)})
\end{displaymath} 
and the like. We refer to the general form of this condition as the
\emph{tensor law}. When we view monads as representations of 
theories, the tensor law just states that the operations of one theory
are homomorphic with respect to\ those of the other theory, i.e.\ it
describes algebras (also known as ``models'') of the first theory in the category of algebras of the
second.


As indicated above, an important example of the tensor product
$\tensor$ is tensoring with global state, in which case the result is
equivalent to application of the state monad
transformer~\cite{PowerShkaravska04}, e.g.\ $\PSet\tensor
T_S=S\to\PSet(S\times\argument)$ where $\PSet$ is the powerset monad
and $T_S=S\to (S\times\argument)$ is the (global) state monad. One can
look at this case from the opposite perspective and consider it as an
application of a \emph{nondeterminism monad transformer}
$T\mapsto\PSet\tensor T$. This transformer yields the universal
\emph{completely additive monad} over
$T$~\cite{GoncharovSchroderEtAl09}, which therefore allows for a
generalized Fischer-Ladner decomposition of control operators, i.e.\
roughly the translation
\begin{align*}
\mathtt{if}~b~\mathtt{then}~p~\mathtt{else}~q 	&:= b?; p + (\neg b)?; q\\
\mathtt{while}~b~\mathtt{do}~p 					&:= (b?;p)^{\varstar};(\neg b)?
\end{align*}
The catch in all this is that unless one requires both component
monads to be \emph{ranked}, i.e.\ generated by a set of
algebraic operations, there is no guarantee that sum and tensor
exist~\cite{HylandPlotkinEtAl06}. Intuitively, unranked monads arise
when the number of values that can participate in a computation is
unbounded.  Unbounded non-determinism and continuations are prominent
examples of unranked monads; in particular, it has long been unclear
that the above-mentioned non-determinism monad transformer actually
exists. It is comparatively easy to see that the sum of simple ranked
monads with most unranked monads will typically fail to exist (see,
e.g.,~\cite{HylandPower07}). The case of tensoring is more subtle. For
the specific example of the (unranked) continuation monad, it has been
shown in~\cite{HylandLevyEtAl07} that the tensor does exist if the
partner monad is ranked. It has been conjectured in \emph{op.\ cit.}
(p.\ 30) that the tensor of unranked monads does not exist in general,
and it has been implicitly indicated (\emph{op.\ cit.}, p.\ 33) that
the tensor of continuations with a suitable unranked monad might serve
as a counterexample, which seemed reasonable insofar as continuations
generally constitute a good source of counterexamples (see, e.g.,
\cite{SchroderMossakowski04}).

However, two of us (Goncharov and Schr\"oder) have recently proved
that the tensor of two monads always exists if one of them is
\emph{uniform}, a natural criterion that ensures sufficiently
pervasive applicability of the tensor
law~\cite{GoncharovSchroder11b}. The class of uniform monads is
surprisingly broad and includes not only countable and unbounded
non-determinism (which implies that the above-mentioned
non-determinism monad transformer does after all exist), but also
continuations, thus discharging the latter as a suspect for a
potential counterexample to existence of tensors. In summary, prior to
the current work (respectively the conference
version~\cite{GoncharovSchroder11}), the question of universal
existence of monad tensors was open, and no good candidates for
possible counterexamples were known. It should be noted that the
question as such dates back at least to~\cite{Manes69}, where it
appears in the context of early developments of the categorical
foundations of universal algebra.

Having said this, we do settle the question in the negative in the
present work. Specifically, we present two countexamples to
tensorability. By the above, at least one of the partner monads in a
counterexample must be unranked, and in both our examples, the other
partner is in fact ranked. One of these, originally presented
in~\cite{GoncharovSchroder11}, shows that the tensor of a
\emph{well-order monad} with a simple free algebra monad (with two
binary operations) fails to exist. In the other example, the ranked
partner is finite powerset, while the unranked partner is less easy to
grasp, being defined by a rather involved equational theory.

Moreover, we settle tensorability of non-empty $\kappa$-bounded
powerset $\pow$ in the positive in the remaining cases. Specifically,
the uniformity method of~\cite{GoncharovSchroder11b} proves
tensorability of $\pow$ for all successor cardinals $\kappa$, and
as mentioned above we show in the present work that tensorability fails for
$\kappa=\omega$. We prove that $\pow$ is tensorable for every
uncountable $\kappa$; the proof method is dedicated to this case and
does not currently seem to generalize to other monads, except that
tensorability of full bounded powerset follows immediately in the
applicable cases.

The paper is organised as follows. We give an introduction to monads
and their algebras as well as their use in programming language
semantics in Section~\ref{sec:monads}. In Section~\ref{sec:tensors},
we discuss tensors and \emph{tensor algebras}, a notion that goes back
to~\cite{Manes69}, and summarize known results (and simple new ones)
on existence of
tensors. 
In Section~\ref{sec:pow-positive}, we prove tensorability of
uncountably bounded powerset. We then proceed to prove our negative
results in Sections~\ref{sec:finpow} and~\ref{sec:non-ex}. In both
counterexamples, tensor algebras appear as a key technical tool in
that we show non-existence of the tensor by exhibiting a family of
reachable tensor algebras of unbounded cardinality.

\subsection*{Note on Foundations}

The category-theoretic concept of a \emph{monad} on $\Set$ serves as an
abstraction of the notion of equational theory. The constructions
involved are well-known in the ranked case; we will discuss the
situation for the unranked case in detail. To formalize the
correspondence between theories and monads, we shall need to
consider theories with large signatures. 

We accordingly work in the von Neumann-Bernays-G\"{o}del (NBG) theory of sets and classes, which is conservative over ZFC set theory~\cite{Felgner:NBGconserve}.  In NBG, certain classes are sets, whilst others, such as the class of sets or the class of ordinals, are proper classes.  All elements of classes are sets.  We shall use ``small'' to indicate that something is a set, and ``large'' to indicate that it might be a proper class.   

We shall also speak of ``hyperlarge'' categories such as the category $\Class$ of classes.   Whereas a familiar large category has a class of objects, a hyperlarge category has a \emph{hyperclass} of  objects, and also a hyperclass of morphisms from one object to another.  Informally, a hyperclass is a collection of classes.   Formally, our statements involving hyperclasses are interpreted at the meta-level, in the same way that ZFC users would interpret statements involving classes: a hyperclass is the extension of a unary formula.  

Note that, in NBG, care is needed to represent quotients and tuples of classes in such a way that they are classes.
\begin{iteMize}{$\bullet$}
\item To quotient a class by an equivalence relation, we may either  use the Global Axiom of Choice (``the class of sets is well-orderable'') to represent each equivalence class by a chosen element, or employ ``Scott's trick''~\cite{Scott55} of representing each equivalence class by its set of elements of least rank.
\item Following Morse~\cite{Morse:sets} an ordered pair $(X,Y)$ of classes is represented as the class $X+Y$. More generally a large (i.e.\ class-indexed) tuple of classes $(X_{i} \mid i \in I)$ is represented as the class $\sum_{i \in I}X_{i}$.  Consequently a hyperlarge sum or large product of hyperclasses is a hyperclass, just as a large sum or small product of classes is a class.
\end{iteMize}

\noindent There are numerous other foundational options besides the one we have chosen, e.g.\ developing a theory of hyperclasses conservative over NBG; interpreting all our statements directly as properties of constructions on ZFC formulas; or using ZFC with one or even two Grothendieck universes, at the cost of losing conservativity over ZFC.  


\section{Monads and  Theories}\label{sec:monads}


\noindent In programming language semantics, monads serve to encapsulate
side-effects, a principle originally due to Moggi~\cite{Moggi91} that
was subsequently introduced into the functional programming language
Haskell as the principal means of dealing with impure
features~\cite{Wadler97}. In a nutshell, the idea is to relocate the
side effect from the function arrow into the result type of a
function: a side-effecting function $X\to Y$ becomes a pure function
$X\to TY$, where $TY$ is a type of side-effecting computations over
$Y$; the base example is $TY=S\to (S\times Y)$ for a fixed set
$S$ of states, so that functions $X\to TY$ are functions that may read
and update a global state.

Formally, a \emph{monad} on $\Set$ consists of a functor
$T:\Set\to\Set$ mapping sets $X$ (of values or, from the point
of view of theories, variables) to sets $TX$ (of
\emph{computations}, or terms modulo equations) and two natural
transformations $\eta:\id\to T$ and $\mu:T^2\to T$, the \emph{unit}
and the \emph{multiplication}, respectively, subject to the equations
$\mu\eta_T=\mu T\eta=\id$ and $\mu T\mu=\mu\mu_T$. We usually denote a
monad by just its functor part $T$, with the other components
understood implicitly.  On $\Set$, any monad has a unique
\emph{strength}, i.e.\ comes equipped with a unique natural
transformation $X\times TY\to T(X\times Y)$ satisfying a number of
conditions~\cite{Moggi91}. A \emph{monad morphism} is a
natural transformation between the underlying monad functors
satisfying obvious conditions with respect to the unit, the
Kleisli extension (equivalently, multiplication) 
see~e.g.~\cite{BarrWells85} for details.  For monads on $\Set$ preservation of strength is automatic.

A monad $T$ on $\Set$ induces two categories, the smallest and the
largest realization of $T$ as an adjunction, respectively: the
\emph{Kleisli category} $\Set_T$ of $T$ has sets as objects and maps
$X\to TY$ as morphisms; the unit $\eta:X\to TX$ serves as the identity
on $X$, and composition is given by \emph{Kleisli composition}
$(f,g)\mapsto \mu (Tf)g$.  On the other hand, the
\emph{Eilenberg-Moore category} $\Set^T$ of $T$ consists of
\emph{$T$-algebras}, which are maps of the form $\alpha:TX\to X$ such
that $\alpha\eta_X=\id_X$ and $\alpha\mu_X=\alpha T\alpha$, and their
morphisms. Here, a morphism $f:\alpha\to\beta$ of $T$-algebras
$\alpha:TX\to X$ and $\beta:TY\to Y$ is a map $f:X\to Y$ such that
$f\alpha=\beta Tf$.

We now turn to theories.  A {\em large signature} $\Sigma$ is a class $\Sigma$ of operations
$f$ each with an associated small arity $\alpha(f)$.  A {\em
  large theory} $\CT = (\Sigma,\CE)$ consists of a large signature
$\Sigma$ and a class $\CE$ of equations between $\Sigma$-terms.  (Formally, an equation is a tuple $(X,t,t')$ where $X$ is a set and $t$ and $t'$ are $\Sigma$-terms over $X$.) An {\em algebra} for such a
theory is a class equipped with a $\Sigma$-structure satisfying all of the
equations in $\CE$.  The algebra is termed \emph{small} if its carrier (not the algebra as a whole) is a set.   For
each set $X$ we build the {\em free algebra} $F_{\CT}(X)$ on $X$ as the class of $\Sigma$-terms over $X$ taken modulo the equations
in $\CT$: this may fail to be small.

More generally, for a category $\CC$ with small products, a {\em $\CT$-algebra} in $\CC$ consists of an object $A$ of $\CC$ together with a map $A^{\alpha(f)} \xrightarrow{f^A} A$ for each $f \in \Sigma$ such that, for each equation $(X,t,t') \in \CE$, the maps $A^{X} \to A$ corresponding to the two terms being equated are equal. An \emph{algebra homomorphism} from $A$ to $B$ consists of a map $A \xrightarrow{k} B$ in $\CC$ commuting with all the operations in the sense that for any $f \in \Sigma$ we have $f^A \cdot  k  ^{\alpha(f)} = k \cdot f$. The category of $\CT$-algebras in $\CC$ is denoted $\CatMod{\CC}{\CT}$. There is an evident forgetful functor from $\CatMod{\CC}{\CT}$ to $\CC$, which creates small products.

Returning to the category of sets, $\CatMod{\Set}{\CT}$ is the
category of small $\CT$-algebras, and there is an evident forgetful functor $U_{\CT}$ from
it to $\Set$. There are two ways in which $U_{\CT}$ can have a
left adjoint (in which case it is even monadic):
\begin{defi}
  A large theory $\CT$ \emph{has small free
    algebras} if $F_\CT(X)$ is small for every set $X$, and
  \emph{free small algebras} if $U_\CT$ has a left adjoint.
\end{defi}
\noindent It is, then, clear that
\begin{enumerate}[(1)]
\item every large theory $\CT$ with small free
  algebras has free small algebras, where the left adjoint of $U_\CT$
  maps $X$ to $F_\CT(X)$;
\item every large theory $\CT$ with free small
  algebras gives rise to, or \emph{presents}, a monad in the standard way, i.e.\ by composing
  $U_\CT$ with its left adjoint,
\end{enumerate}
i.e.\ we have
\begin{equation}
  \label{eqn:sffs}
  \begin{array}{ccccc}
    \begin{array}{c}
\text{Large theories}\\\text{with small free algebras}
\end{array}
 & \subseteq &
 \begin{array}{c}
\text{Large theories}\\\text{with free small algebras}
\end{array}
 & \to & 
\text{Monads on $\Set$}
\end{array}
\end{equation}
where the arrow $\to$ is intended to denote a map. We make three observations about (\ref{eqn:sffs}).  Firstly the inclusion is strict.
\begin{thm} \label{thm:sffs}
  There exists a large theory that has free small algebras but not small free algebras.
\end{thm}
\begin{proof}
   Let $\CT$ be the following theory.  For the signature we take a constant $k_\alpha$
  for each ordinal $\alpha$ and a function $f$ of arity 3.  For the equations we take  $f(k_{\alpha}, k_{\alpha}, x) = k_0$
  for any $\alpha$, and $f(k_{\alpha}, k_{\beta}, x) = x$ for
  any distinct $\alpha,\beta$. Then the elements of
  $F_{\CT}(\emptyset)$ are in bijection with the set of constants
  $k_{\alpha}$, with the action of $f$ completely determined by the
   equations in $\CT$, so $F_{\CT}(\emptyset)$ fails to be small, i.e.\
  $\CT$ does not have small free algebras. On the other hand,
  in any small $\CT$-algebra $A$, the interpretations of $k_{\alpha}$
  and $k_{\beta}$ must be equal for some pair of distinct ordinals
  $\alpha,\beta$. Then for any $x$ in $A$ we have $x =
  f(k_{\alpha}, k_{\beta}, x) = f(k_{\alpha}, k_{\alpha}, x) = k_0$,
 so that $A$ has only one element. Hence $U_{\CT}$
  has a left adjoint, i.e.\ $\CT$ has free small algebras.
\end{proof}
\noindent Suppose $\CT$ is a theory with free small algebras that presents the monad $\Amonad$.   We call $\CT$ a \emph{genuinely presenting theory} of $\Amonad$ if $\CT$ has small free algebras, and otherwise a \emph{spuriously presenting theory}.  Thus the theory in the preceding proof spuriously presents the monad $X \mapsto 1$.  We next see that every monad has a genuinely presenting theory.
\begin{defi}
  Let $(T,\eta,-^{*})$ be a monad on $\Set$ expressed as a Kleisli triple.  We define the large theory $\Theorify{\Amonad}$ as follows.   For the signature, we take for every set $X$ and $m \in TX$ an operation $h_{X,m}$ of arity $X$.  For the equations, we take
  \begin{iteMize}{$\bullet$}
  \item for every set $X$ and $a\in X$, an equation
    \begin{equation*}
      h_{X, \eta_X(a)} (p_x \mid x\in X) = p_a.
    \end{equation*}
  \item for all sets $X,Y$ and $m\in TX,f : X \to TY$, an equation
 \begin{equation*}
     h_{X, m} (h_{Y,f(x)} (p_y \mid y\in Y) \mid x\in X ) = h_{Y, f^{*}m}(p_y \mid y\in Y).
    \end{equation*} 
  \end{iteMize}
\end{defi}

\begin{thm} \label{thm:montheory}
Any monad $\Amonad$ on $\Set$ is genuinely presented by the large theory $\Theorify{\Amonad}$.  Thus we have
\begin{displaymath}
   \Set^{\Amonad}  \congsp{\Set}  \CatMod{\Set}{\Theorify{\Amonad}} 
  \end{displaymath}
where $\cong_{\Set}$ indicates an isomorphism commuting with the forgetful functors to $\Set$.
\end{thm}
\noindent We exploit the theory $\Theorify{\Amonad}$ to define algebras in other categories.
\begin{defi} \label{def:alggen}
Let $\Amonad$ be a monad on $\Set$.
  For any category with small products $\CC$, we define the \emph{$\Amonad$-algebras in $\CC$} via
  \begin{displaymath}
    \CatMod{\CC}{\Amonad} \eqdefsp \CatMod{\CC}{\Theorify{\Amonad}}
  \end{displaymath}
\end{defi}

\begin{rem}
$\CatMod{\CC}{\Amonad}$ is equivalent to the category of functors from $\Set_{\Amonad}^{\op}$ to $\CC$ preserving small products.  The latter was used for the same purpose in~\cite{Dubuc70,HylandLevyEtAl07}.
\end{rem}

\noindent We close the circle by seeing that the composite (\ref{eqn:sffs}) does not affect the category of algebras.
\begin{thm} \label{thm:alggood}
Let $\Amonad$ be a monad on $\Set$ genuinely presented by $\CT$.   For any category $\CC$ with small products, we have
  \begin{displaymath}
    \CatMod{\CC}{\CT} \congsp{\CC} \CatMod{\CC}{\Amonad}
  \end{displaymath}
\end{thm}
\noindent Theorem~\ref{thm:alggood} tells us that the category of algebras depends only on the monad, not on the choice of genuinely presenting theory.   With these conversions in place, we will switch back and forth freely
between monads and large theories as convenient. (Monads are
also formally equivalent to large Lawvere theories~\cite{Linton66},
which were used in the theory of generic side-effects
in~\cite{HylandLevyEtAl07}.)

The following easy results will be useful.
\begin{prop} \label{lemma:adapttheory}\hfill
  \begin{enumerate}[\em(1)]
  \item Let $T_0$ be a monad on $\Set$ genuinely presented by $\CT = (\Sigma,\CE)$, and $\alpha : T_0 \to T_1$ a componentwise surjective monad morphism.  Then $T_1$ is genuinely presented by the theory with signature $\Sigma$ and equation class
    \begin{displaymath}
      \CE \cup \{ (X,t,t') \mid \alpha_X[t]_{\CE} = \alpha_X[t']_{\CE}  \}.
    \end{displaymath}
\item Let $T$ be a monad on $\Set$ genuinely presented by $\CT = (\Sigma,\CE)$. For any set $E$, the monad $T(- + E)$ is genuinely presented by the theory with signature $\Sigma$ extended by  a family of constants $(c_e \mid e \in E)$, and equation class $\CE$.
  \end{enumerate}
\end{prop}

\noindent The notion of \emph{rank} for a monad refers to the arity of
the involved algebraic operations. Formally, a monad is
\emph{$\kappa$-ranked} for a regular cardinal $\kappa$ if the
underlying functor preserves $\kappa$-filtered colimits. A monad is
\emph{ranked} if it is $\kappa$-ranked for some $\kappa$. The
$\kappa$-ranked monads on $\Set$ are precisely those that are induced
by a \emph{small} theory in which the arity of each operation is less
than $\kappa$, so that the ranked monads are precisely those induced
by small theories. Computationally relevant unranked monads include
the continuation monad and the unbounded powerset and are presented
further below. It has been shown that the algebraic view of ranked
monads gives rise to computationally natural operations; e.g.\ the
state monad (with state set $S=V^L$ for sets $V$ of values and $L$ of
locations) can be algebraically presented in terms of operations
$\mathit{lookup}$ and $\mathit{update}$~\cite{PlotkinPower02}.

We give some standard examples of computational monads, mostly
from~\cite{Moggi91}:
\begin{exa}[Computational Monads]\label{expl:theories}
  \mbox{}\begin{enumerate}[(1)]
  \item \emph{Global state:} as stated initially, $TX=S\to(S\times X)$
    is a monad (for this and other standard examples, we omit the
    description of the remaining data), the well-known \emph{state
      monad}. 
  \item\label{item:non-det} \emph{Nondeterminism:} the generic (unranked)
    monad for nondeterminism is the one presented by the covariant powerset functor $\PSet$. Variants arise on the one hand by restricting
    to nonempty subsets, thus ruling out non-termination, and on the
    other hand by bounding the cardinality of subsets. We denote
    nonemptyness by a superscript $\varstar$, and cardinality bounds
    by subscripts. E.g., the monad ${\PSet^{\varstar}_{\omega_1}}$ describes countable non-blocking
    nondeterminism.
    Yet another variant arises by replacing sets with countable multisets, i.e.\
    maps $X\to(\Nat\cup\{\infty\})$, thus modelling weighted
    nondeterminism~\cite{DrosteKuichEtAl09}. Let us denote by $T_{\mathit{mult}}$ the corresponding countable multiset monad.
  \item\label{item:cont} \emph{Continuations:} The continuation monad
    maps a set $X$ to the set $(X\to R)\to R$, for a fixed set $R$ of
    results. We denote the corresponding unranked monad as $T^R_{\mathit{cont}}$.
  \item \emph{Input/Output:} For a given set $I$ of input symbols, the
    monad $T_I$ for input is generated by a single $I$-ary operation;
    this monad is induced by an \emph{absolutely free} theory, i.e.\
     one without equations. Similarly, given a set $O$ of output
    symbols, the monad $T_O$ for output is induced by a family of
    unary operations indexed over $O$.
  \end{enumerate}
\end{exa}

\noindent A convenient way of denoting generic computations is the
so-called \emph{computational metalanguage}~\cite{Moggi91},
which has found its way into functional programming in the shape of
Haskell's do-notation. We briefly outline the version of the
metalanguage we use below; this version is deliberately simplistic, as
it serves only to elucidate the definition of tensors.

The metalanguage denotes morphisms in the underlying category of a
given monad, using the monadic structure; since we are working over
$\Set$, the metalanguage just denotes maps in our setting. We let a
signature $\Sigma$ consist of a set $\mathcal{B}$ of \emph{base
  types}, to be interpreted as sets, and a collection of typed
\emph{function symbols} $f:A_1\to A_2$ to be interpreted as functions,
where $A_1,A_2$ are types. Here, we assume that the set $\mathcal{T}$
of \emph{types} is generated from the base types by the grammar
\begin{equation*}
  \mathcal{T}\owns A_1,A_2::=1\mid B\mid A_1\times A_2\mid TA_1\qquad (B\in\mathcal{B})
\end{equation*}
where $\times$ is interpreted as set theoretic product, $1$ is a
singleton set, and $T$ is application of the given monad. We then have
standard formation rules for terms-in-context $\Gamma\rhd t:A$, read
`term $t$ has type $A$ in context $\Gamma$', where a \emph{context} is
a list $\Gamma=(x_1:A_1,\dots,x_n:A_n)$ of typed variables (later,
contexts will mostly be omitted):

\begin{align*}
&&\infrule{x:A\in\Gamma}{\Gamma \rhd x:A}&&\quad
\infrule{f:A\to B\in\Sigma\qquad\Gamma\rhd t:A}{\Gamma\rhd f(t):B}&&\quad
\infrule{}{\Gamma\rhd\unit:1}&&
\end{align*}

\begin{align*}
&&\infrule{\Gamma\rhd s:A\qquad\Gamma\rhd t:B}
{\Gamma\rhd \langle s,t\rangle : A \times B}&&\quad
\infrule{\Gamma\rhd t:A\times B}
{\Gamma\rhd\fst t:A}&&\quad
\infrule{\Gamma\rhd t:A\times B}
{\Gamma\rhd\snd t:B}&&
\end{align*}

\begin{align*}
&&
\infrule{\Gamma \rhd t:A}
{\Gamma \rhd \ret{t}:TA}&&\quad
\infrule{\Gamma\rhd p:TA\qquad\Gamma, x:A\rhd q:TB}
{\Gamma\rhd \letTerm{\bind{x}{p}}{q}:TB}&&\\[.01ex]
\end{align*}

\noindent Only the operations in the last line are specific to monads;
they are called~\emph{return} and~\emph{binding}, respectively. For
binding, we use 
Haskell's $\DO$-notation
. Return is interpreted by the unit $\eta$ of the monad, and
can be thought of as returning a value. A binding
$\letTerm{\bind{x}{p}}{q}$ 
executes $p$, binds its result to $x$, and
then executes $q$, which may use $x$
. Binding is right associative, i.e.\
$\letTerm{\bind{x}{p}}{\letTermR{\bind{y}{q}}{r}}=\letTerm{\bind{x}{p}}{(\letTerm{\bind{y}{q}}{r})}$.
It is interpreted using Kleisli composition and strength, where the
latter serves to propagate the context $\Gamma$~\cite{Moggi91}. In
consequence, one has the \emph{monad laws}
\begin{gather*}
\letTerm{\bind{x}{p}}{\ret{x}} = p\qquad
\letTerm{\bind{x}{\ret{a}}}{p} = p[a/x]\\[1ex]
\letTerm{\bind{x}{(\letTerm{\bind{y}{p}}{q})}}{r} = \letTerm{\bind{y}{p}}{\letTermR{\bind{x}{q}}{r}}.
\end{gather*}
Terms of a type $T A$ are called~\emph{programs}. We say that two programs $p:TA$ and $q:TB$ \emph{commute} if they satisfy the equation
\begin{equation}\label{eq:comm}
\letTerm{\bind{x}{p}}{\letTermR{\bind{y}{q}}{\ret\brks{x,y}}} = \letTerm{\bind{y}{q}}{\letTermR{\bind{x}{p}}{\ret\brks{x,y}}} \ : \ T(A\times B).
\end{equation}

\begin{rem}
  The notion of commutation of programs relates as expected to the
  standard notion of \emph{commutative monad}: a monad is commutative
  iff all its programs commute.
\end{rem}

\section{Tensors and Tensor Algebras}\label{sec:tensors}

\noindent We begin by defining tensor of monads in terms of a universal property, just as the coproduct is defined to be an initial cocone.
\begin{defi} \label{def:univtensor}
    Let $\Monadone$ and $\Monadtwo$ be monads on $\Set$.
     A cocone from $\Monadone$ and $\Monadtwo$----that is to say, a
      monad $R$ together with morphisms $\tau:\Monadone\to R$
      and $\sigma:\Monadtwo\to R$---is a \emph{commuting cocone}
      when every two programs of the form $\tau_X(p)$ and
      $\sigma_Y(q)$ commute (see Section~\ref{sec:monads}).    A \emph{tensor product} $\Monadone \tensor \Monadtwo$ is an
      initial commuting cocone, i.e.\ a commuting cocone from which there is a unique cocone morphism to any commuting cocone.
\end{defi}
\noindent Just as the coproduct of monads corresponds to the disjoint union of theories~\cite{AdamekBowlerLevyMilius:coprodmonset,HylandPlotkinEtAl06,Kelly:transfin,LuthGhani02}, so the tensor can be described in terms of theories.

\begin{defi}
  Let $\CTone$ and $\CTtwo$ be large theories with signatures
  $\Sigma_1$ and $\Sigma_2$ respectively.
  \begin{enumerate}[(1)]\item 
    Their {\em disjoint union} has as its signature the disjoint union
    of $\Sigma_1$ and $\Sigma_2$, and its equations consist of those
    of $\CTone$ together with $\CTtwo$. 
\item The {\em tensor product} $\CTone
    \tensor \CTtwo$ of the theories is on the same signature as the
    disjoint union, and has as its equations all those of the
    disjoint union together with the {\em commutativity} equation
    \begin{equation}\label{eq:commute}
    f(g(x_{ij}\mid j \in \alpha(g))\mid i \in \alpha(f)) = g(f(x_{ij}\mid i \in
    \alpha(f))\mid j \in \alpha(g))
  \end{equation}
saying that $f$ and $g$ {\em commute}
    for any $f \in \Sigma_1$ and $g \in \Sigma_2$.
  \end{enumerate}
\end{defi}
\begin{rem}
  A theory $\CT$ with signature $\Sigma$ is said to be
  \emph{commutative} if all its operations are algebra homomorphisms,
  that is, if Equation~(\ref{eq:commute}) holds for all
  $f,g\in\Sigma$. Clearly, if $\CT$ has small free
  algebras, then $\CT$ is commutative iff the monad it presents is
  commutative.
\end{rem}
\noindent Tensor algebras can also be viewed as algebras in categories
of algebras.
\begin{prop}
For large theories $\CTone$ and $\CTtwo$ and category $\CC$ with small products, we have
\begin{displaymath}
  \CatMod{\CC}{(\CTone \tensor \CTtwo)}
\congsp{\CC} \CatMod{\CatMod{\CC}{\CTtwo}}{\CTone}
\end{displaymath}
\end{prop}
\noindent As in Section~\ref{sec:monads} we adapt tensor algebras from theories to monads.
\begin{defi} \label{def:tensoralg}
  Let $\Monadone$ and $\Monadtwo$ be monads on $\Set$.
  \begin{enumerate}[(1)]
  \item~\cite{Manes69} A \emph{small $(\Monadone,\Monadtwo)$-tensor algebra} is a triple $(X,\alpha,\beta)$ where $X$ is a set and $\alpha$ and $\beta$ are respectively Eilenberg-Moore $\Monadone$- and $\Monadtwo$-algebra structures on $X$, such that for all sets $Y,Z$ and all $p\in S Y$, $q\in T Z$,
  $f:Y\times Z\to X$, the following equation, called the \emph{tensor
    law}, holds
  \begin{flalign*}
    \beta(\Monadone(\lambda z.\, \alpha (\Monadtwo f_{\_,z}\; p)) q) = \alpha(\Monadtwo(\lambda
    y.\, \beta (T f_{y,\_}\; q)) p)
  \end{flalign*}
  where $f_{\_,z}(y)=f_{y,\_}(z)=f(y,z)$ for $(y,z)\in Y\times Z$.
  Morphisms of $(\Monadone,\Monadtwo)$-tensor algebras are maps between the respective
  carriers which are homomorphic for both $\Monadone$ and $\Monadtwo$. 
 The (large, locally small) category
  of $(\Monadone,\Monadtwo)$-tensor algebras is denoted $\SmallTalgOf{\Monadone}{\Monadtwo}$.
  \item For any category with small products $\CC$, we define
    \begin{displaymath}
      \TAlgOf{\CC}{\Monadone}{\Monadtwo} \eqdefsp \CatMod{\CC}{(\Theorify{\Monadone}\tensor \Theorify{\Monadtwo})}
    \end{displaymath}
  \end{enumerate}
\end{defi}
\noindent We then have the following counterpart of Theorems~\ref{thm:montheory} and \ref{thm:alggood}.
\begin{thm}\label{thm:nomattertensor}
  Let $\Monadone$ and $\Monadtwo$ be monads on $\Set$.
  \begin{enumerate}[\em(1)]
  \item The two parts of Def.~\ref{def:tensoralg} agree: we have 
    \begin{displaymath}
      \SmallTalgOf{\Monadone}{\Monadtwo} \congsp{\Set} \TAlgOf{\Set}{\Monadone}{\Monadtwo}
    \end{displaymath}
  \item \label{item:nomattertensor} Let $\CTone$ and $\CTtwo$ be genuinely presenting theories for $\Monadone$ and $\Monadtwo$ respectively.  For any category $\CC$ with small products, we have
    \begin{displaymath}
      \CatMod{\CC}{(\CTone \tensor \CTtwo)} \congsp{\CC} \TAlgOf{\CC}{\Monadone}{\Monadtwo}
    \end{displaymath}
  \end{enumerate}
\end{thm}
\noindent Once again part (\ref{item:nomattertensor}) tells us that the category of tensor algebras depends only on the two monads, not on the choice of genuinely presenting theories. 

Next we relate tensor algebras to the universal property of Def.~\ref{def:univtensor}.  The following result appears
 (modulo translation from the language of large Lawvere theories into the language of monads) in~\cite{HylandLevyEtAl07}.
\begin{thm}\label{prop:tensor-ex}   
 The tensor product of monads $\Monadone, \Monadtwo$ on $\Set$ exists if and only if the
  forgetful functor from $(\Monadone,\Monadtwo)\TAlg$ to $\Set$ is monadic,
  equivalently has a left adjoint. In this case, the monad induced by
  the adjunction is the tensor $\Monadone\tensor \Monadtwo$.
\end{thm}



\begin{remdefn}\label{rem:spurious}
  There are essentially three alternatives regarding the existence of
  a tensor of monads $\Monadone,\Monadtwo$ genuinely presented by large
   theories $\CTone$, $\CTtwo$.
  \begin{enumerate}[(1)]\item 
    The tensor theory $\CTone\tensor\CTtwo$ may fail to have
    free small algebras, in which case $\Monadone\tensor \Monadtwo$
    does not exist. 
\item $\CTone\tensor\CTtwo$ may have free small
    algebras but not small free algebras, in which case the tensor
    $\Monadone\tensor \Monadtwo$ does exist but does not have the
    expected form; we call such monad tensors
    \emph{spurious}. 
\item $\CTone\tensor\CTtwo$ may have small
    free algebras, in which case $\Monadone\tensor \Monadtwo$ exists
    and has the expected form, i.e.\ maps a set $X$ to the underlying
    set of $F_{\CTone\tensor\CTtwo}(X)$; we call such monad tensors
    \emph{genuine}.
  \end{enumerate}
By Theorem \ref{thm:nomattertensor}(\ref{item:nomattertensor}) with $\CC=\Class$, this three-way classification must depend only on the monads $\Monadone$ and $\Monadtwo$ and not on the choice of genuinely presenting theories.  We shall say that a monad $\Monadone$ on $\Set$ is \emph{tensorable} when $\Monadone \tensor \Monadtwo$ exists for every monad $\Monadtwo$ on $\Set$.  If these tensors are all genuine, $\Monadone$ is \emph{genuinely tensorable}.

It is currently an open
  question whether spurious monad tensors exist.  Our results
  established below are always the stronger of the two possible
  variants: where we show existence of tensors, we actually show also
  that the tensor is genuine, and where we show non-existence, we
  prove that not even a spurious tensor exists.
\end{remdefn}

The following is a straightforward consequence of Prop.~\ref{lemma:adapttheory}.
\begin{prop}\label{prop:tensor-inheritance}
  \begin{enumerate}[\em(1)]
  \item Let $T_0,T_1,T_2$ be monads on $\Set$, and $\alpha : T_0 \to T_1$  a componentwise surjective monad morphism.  If $T_0 \tensor S$ exists then so does $T_1 \tensor S$ and the induced morphism $T_0 \tensor S \to T_1 \tensor S$ is componentwise surjective.  Moreover if $T_0 \tensor S$ is genuine then so is $T_1 \tensor S$. 
  \item Let $T$ and $S$ be monads on $\Set$, and $E$ a set.  If $T \tensor S$ exists then so does $T(- + E) \tensor S$ and the induced morphism $(T\tensor S)(- + E) \to T(- + E) \tensor S$ is componentwise surjective.  Moreover if $T \tensor S$ is genuine then so is $T(- + E) \tensor S$. 
  \end{enumerate}
\end{prop}

\noindent
Our negative results will be based on the following simple result.
\begin{defi}
As usual, we say that an algebra is \emph{generated} by a subset $X$
if it does not have a proper subalgebra containing $X$. We say that an
algebra is \emph{$\alpha$-reachable} for a cardinal $\alpha$ if it has
a generating set $X$ of cardinality $|X|\leq\alpha$.
\end{defi}
\begin{cor}\label{cor:tensor-initial}
  The (possibly spurious) tensor $\Monadone\tensor \Monadtwo$ of monads $\Monadone,\Monadtwo$ on $\Set$ exists if and only
  if for every cardinal $\alpha$, the cardinality of
  $\alpha$-reachable small $(\Monadone,\Monadtwo)$-tensor algebras is bounded.
\end{cor}
\begin{proof}
  \emph{`Only if'}: If the tensor exists, then there is a free small
  $(\Monadone,\Monadtwo)$-tensor algebra over $\alpha$, and every $\alpha$-reachable small $(\Monadone,\Monadtwo)$-tensor algebra
  is a quotient of it.

  \emph{`If':} To show the forgetful functor from $(\Monadone,\Monadtwo)$-tensor algebras to
  $\Set$ has a left adjoint under the given assumption, we apply Freyd's general adjoint functor
  theorem, since every set carries only set-many small
  $(T,S)$-tensor algebras.
\end{proof}

\noindent Prior to the current results, the state of research
regarding counterexamples to tensorability was as follows. It is
well-known that the tensor of two ranked monads on $\Set$ does
exist~\cite{HylandLevyEtAl07}, so that any counterexample needs to
involve at least one unranked monad. One unranked monad that is known
to show hard-to-control behaviour in many respects is the continuation
monad. It has been shown in~\cite{HylandLevyEtAl07} that the tensor of
any ranked monad with the continuation monad exists, and at the same
time it has been conjectured that the continuation monad fails to be
tensorable (i.e.\ that there exists an unranked monad whose tensor
with the continuation monad fails to exist). However, it has
subsequently been shown that all so-called \emph{uniform} monads are
tensorable~\cite{GoncharovSchroder11b}, and the given proof in fact
implies that the tensors in question are genuine. The class of uniform
monads is quite broad and in particular includes both the powerset
monad and the continuation monad, so that these monads are ruled out
as counterexamples to tensorability.

\section{Tensoring With Bounded Powerset Monads}\label{sec:pow-positive}

\noindent As discussed in the last section, it has been shown using
the uniformity method that the unbounded and countable powerset monads
as well as their restrictions to non-empty subsets are
tensorable~\cite{GoncharovSchroder11b}, and we will show in
Section~\ref{sec:finpow} that the finite powerset monad fails to be
tensorable. Right now, we will show that all non-empty uncountably
bounded powerset monads, i.e.\ submonads of the powerset monad of the
form $\pow$, where $\pow(X)$ denotes the set of non-empty subsets of
$X$ of cardinality less than $\kappa$ for an uncountable regular
cardinal $\kappa$, are genuinely tensorable. (Regularity of $\kappa$
is equivalent to $\pow$ actually being a monad. Requiring $\kappa$ to
be uncountable ensures that $\pow$-algebras have countable joins.)
Interestingly, the proof does not seem to relate to any generalization
of uniformity. From genuine tensorability of $\pow$, genuine
tensorability of the full bounded powerset monad $\PSet_{\kappa}$
(which maps a set $X$ to the set of \emph{all} subsets of $X$ of
cardinality less than $\kappa$) is immediate by
Proposition~\ref{prop:tensor-inheritance} (alternatively, the
proof below can be adapted to full bounded powerset, and in fact
becomes simpler in the process).

Let $\Amonad$ be a monad genuinely presented by $\CT$, and let $X$ be a set. Let $A$ be the free large
$(\pow,\Amonad)$-tensor algebra on $X$; we need to show that $A$ is
small. Let $\hat X$ be the image of $X$ in
$A$. 
For any subset $Y$ of $A$, let $T``Y$ denote the (necessarily small) sub-$\Amonad$-algebra of
$A$ generated by $Y$, and similarly for $\pow$. In general a $\pow$-algebra is a semilattice in which every nonempty subset of size $<\kappa$ has a supremum; we denote by $\le$ the
ordering on $A$ induced by the $\pow$-structure.

\begin{lem}\label{claim:t}
For any $x \in A$ there is $t \in T``{\hat X}$ such that $x \geq t$.
\end{lem}
\begin{proof}
By induction on the complexity of terms.
For any $x \in \hat X$, we have $x \geq x\in T``{\hat X}$.

Next, let $x$ have the form $\bigvee_{i \in I} x_i$, where $I$ is a nonempty set of cardinality smaller than $\kappa$, and pick $i\in I$. By the induction hypothesis, we have some $t \in T``{\hat X}$ with $t \leq x_i \leq x$.

Finally, let $x$ have the form $f(x_j\mid j \in \alpha(f))$ for some operation $f$ of $\CT$. By the induction hypothesis we can pick $(t_j \in T``{\hat X}\mid j \in\alpha(f))$ with $t_j \leq x_j$ for each $j \in\alpha(f)$. Then $f(x_j\mid j \in\alpha(f)) = f(x_j \vee t_j\mid j \in\alpha(f)) = f(x_j\mid j\in\alpha(f)) \vee f(t_j\mid j \in\alpha(f))$, so that $x\geq f(t_j\mid j\in\alpha(f)) \in T``T``{\hat X} = T``{\hat X}$.
\end{proof}

\noindent Now define a sequence of subsets $(X_n\mid  n \in {\mathbb N})$ of $A$ by $X_0 = {\hat X}$ and $X_{n+1} = \pow``T``X_n$, and let $X_{\omega} = \bigcup_{n \in \mathbb N} X_n$.

\begin{lem}\label{claim:sup}
Any $x \in \pow``X_{\omega}$ can be written in the form $\bigvee_{n \in \Nbb} x_n$, with each $x_n \in X_n$
\end{lem}
\begin{proof}
We can write $x$ in the form $\bigvee_{i \in I} x_i$ with $I$ nonempty and of cardinality $<\kappa$ and  $x_i \in X_{\omega}$ for all $i$. For $n \in \Nbb$, let $I_n = \{i \in I \mid x_i \in X_n\}$, so that $\bigcup_{n \in \Nbb} I_n = I$. Pick (by Lemma~\ref{claim:t}) $t \in T``{\hat X}$ such that $t \leq x$, and $n_0 \in \Nbb$ minimal such that $I_n$ is nonempty. Taking $x_n = t\in T``{\hat X}\subseteq X_n$ for $0<n < n_0$ and $x_n=\bigvee_{i \in I_n} x_i$ for $n\ge n_0$, we obtain $x = \bigvee_{n \in \Nbb} x_n$, as required.
\end{proof}

\begin{lem}
$T``\pow `` X_{\omega} \subseteq \pow ``X_{\omega}$
\end{lem}
\begin{proof}
By Lemma~\ref{claim:sup}, any $x \in T``\pow``X_{\omega}$ can be expressed as $f\left(\bigvee_{n \in \Nbb} x_{i,n}\mid i \in \alpha(f)\right)$ for some operation $f$ of $T$ and $x_{i,n} \in X_n$. Commutation of the $\pow$-structure with $f$ implies $x = \bigvee_{n \in \Nbb} f(x_{i,n}\mid i \in \alpha(f)) \in \pow `` X_{\omega}$ since for each $n \in \Nbb$ we have $f(x_{i,n}\mid i \in \alpha(f)) \in T``X_n \subseteq X_{n+1} \subseteq X_{\omega}$.
\end{proof}

\noindent It follows that $A = T``\pow `` X_{\omega}$, and in
particular $A$ is small. Since $X$ was arbitrary, this implies that
the monad tensor $\pow \tensor T$ exists, and since $T$ was arbitrary,
we obtain (using Proposition~\ref{prop:tensor-inheritance} again)

\begin{thm}
  For every regular cardinal $\kappa>\omega$, $\pow$ and
  $\PSet_\kappa$ are genuinely tensorable.
\end{thm}

\section{Finite Powerset Fails to be Tensorable}
\label{sec:finpow}
\noindent We now turn to our negative results, i.e.\ examples of two
monads whose tensor fails to exist. Necessarily, one of these must be
unranked. In this section and the next, we present two examples of
this kind. The first, to be discussed presently, involves the finite
powerset monad, which is computationally significant as a monadic
model of finite non-determinism; the unranked partner in the example
is a somewhat involved theory that we explain in detail below. The
second example involves to comparatively simple monads, a type of
well-order monad and a free monad; this was actually the first example
to be found~\cite{GoncharovSchroder11}.

In both cases, we construct the unranked partner monad via a large
 theory. First, we introduce the theory and show
that has small free algebras, and hence induces a monad. Then we show
that the tensor product with the (small) theory of the
ranked partner has arbitrarily large tensor algebras all generated by
some particular set. The fact that this tensor product does not
induce a monad is then immediate from Corollary
\ref{cor:tensor-initial}.

We shall think of $\PSet_{\omega}$ as the `free bounded semilattice' monad, that is as the monad corresponding to the  theory of bounded semilattices.


We now define the unranked partner monad $\CM$ as induced by a large
 theory $\CT_{\CM}$. The signature $\Sigma$ of $\CT_{\CM}$ is defined as follows.
For any ordinal $\alpha$, we define $\alpha^+$ to be the least
regular ordinal greater than $\alpha$. Given a ordinal $\delta$,
let $B_{\delta}$ be the set of all triples $(\alpha, \beta, i)$
of ordinals such that $\alpha<\delta$, $\beta < \alpha^+$, $i <
\omega$.
%
Let $\Sigma_{\kappa}$ be the signature consisting of a constant symbol $c$ and operation symbols $f_{\alpha, \beta, i}$ of arity $B_{\alpha}$ for each $(\alpha, \beta, i) \in B_{\kappa + 1}$, $\alpha>0$ and let $\Sigma$ be the (large) union of all the $\Sigma_{\kappa}$. 

The theory $\CT_{\CM}$ consists of equations 
\begin{equation}
f_{\alpha, \beta, i}(t_{\alpha',\beta', i'} \mid (\alpha', \beta', i') \in B_{\alpha})=c
\end{equation}
whenever the $t_{\alpha', \beta', i'}$ are terms in variables drawn from some set $X$ of cardinality strictly smaller than that of $\alpha$.

We may understand algebras for this theory as follows. 
Given a $\Sigma$-algebra $A$, a set $X$, a map $x\colon X\to A$, and
$\Sigma'\subseteq\Sigma$, let us denote by $\brks{x}_{\Sigma'}$ the
$\Sigma'$-subalgebra of $A$ generated by the image of $x$. Then a
$\CT_{\CM}$-algebra is just a $\Sigma$-algebra that satisfies the
following property: Given a cardinal $\kappa$, a set $X$ of
cardinality $\kappa$ and a map $x:X\to A$, the equation
\begin{equation}\label{eq:m-alg}
f_{\alpha, \beta, i}(a_{\alpha',\beta', i'} \mid (\alpha', \beta', i') \in B_{\alpha})=c
\end{equation}
holds over $\brks{x}_{\Sigma_{\kappa}}$ whenever $\alpha>\kappa$. Note
that this implies
\begin{equation}
 \langle x \rangle_{\Sigma} = \langle x \rangle
_{\Sigma_{\kappa}}.\label{eq:m-alg-gen}
\end{equation}


\begin{lem}
  The theory $\CT_{\CM}$ has small free algebras.
\end{lem}
\begin{proof}
  The free algebra over a set $X$ satisfies~\eqref{eq:m-alg-gen}
%
  for $\kappa=|X|$. 
\end{proof}
\noindent We denote the monad induced by $\CT_\CM$ by the above lemma
by $\CM$. We proceed to construct a sequence of reachable
$\CT_{\CM}$-algebras of unbounded cardinality.  For any cardinal
$\delta$, put $M_{\delta} = \PSet_{\omega}(B_{\delta})$. For the
definition of the $\Sigma$-algebra structure, we need a few technical
preliminaries.
\begin{defi}\label{subsume}
Let $X$ be a set, and let $(x_i \mid i \in \omega)$ be a sequence of elements of $X$. We say that a sequence $(a_i \mid i \in \omega)$ of finite subsets of $X$ {\em subsumes} $(x_i \mid i \in \omega)$ if and only if there is an infinite subset $M$ of $\omega$ such that for any $i, j \in M$ with $i < j$ we have $x_i \in a_j$. 
\end{defi}
\noindent It will be helpful to make use of an unusual sort of quantifier: we use $\exists_{\cofinal} \beta < \alpha.\,\phi(\beta)$ to mean `the set of $\beta < \alpha$ such that $\phi(\beta)$ holds is cofinal in $\alpha$'\footnote{Recall that a subset $C$ of $\alpha$ is cofinal in $\alpha$ if and only if $\forall \beta < \alpha.~\exists \gamma \in C.~\beta \leq \gamma$}. %
\noindent  We put a $\Sigma$-algebra structure on $M_{\delta}$ by taking $f_{\alpha, \beta, i}(a_{\alpha',\beta', i'} \mid (\alpha', \beta', i') \in B_{\alpha})$ to be 
\begin{equation*}
  \begin{array}{rl}
\{(\alpha, \beta, i)\}  &\text{if $\alpha<\delta$ and }\\
&\exists_{\cofinal} \alpha' < \alpha.\,\exists_{\cofinal} \beta' < \alpha'^+.\,(a_{\alpha', \beta', i'} \mid i' \in \omega) \text{ subsumes } ((\alpha', \beta', i')\mid i' \in \omega) \\ \emptyset  &\text{otherwise}
\end{array}
\end{equation*}
and interpreting $c$ as the empty set. We proceed to show that this
does indeed define an $\CT_{\CM}$-algebra.
\begin{lem}\label{above1}
If $(a_i\mid i \in \omega)$ subsumes a sequence $(x_i\mid i \in \omega)$ of distinct elements of $X$ then some $a_i$ contains at least two elements of that sequence. \easy
\end{lem}

%
%
%

\noindent 
%
%
\begin{lem}\label{lem:neg2}
For every $\delta$, $M_\delta$ is a $\CT_{\CM}$-algebra.
\end{lem}
\begin{proof}


\noindent 
Suppose for a contradiction that $f_{\alpha, \beta, i}(a_{\alpha', \beta', i'} \mid (\alpha', \beta', i') \in B_{\alpha}) \neq c=\emptyset$ with $a_{\alpha',\beta',i'}\in\langle x\rangle_{\Sigma^\kappa}$ for some $x:X\to M_\delta$ where $|X|=\kappa$ and $\alpha>\kappa$. Then ($\alpha<\delta$ and)
\begin{displaymath}
\exists_{\cofinal} \alpha' < \alpha.~\exists_{\cofinal} \beta' < \alpha'^+.~ (a_{\alpha', \beta', i'} \mid i' \in \omega) \text{ subsumes } ((\alpha', \beta', i')\mid i' \in \omega).
\end{displaymath}
Pick $\alpha_0$ such that $\kappa\leq\alpha_0<\alpha$ and the set
\begin{displaymath}
C := \{\beta' < \alpha_0^+ \mid (a_{\alpha_0, \beta', i'} \mid i' \in \omega) \text{ subsumes } ((\alpha_0, \beta', i')\mid i' \in \omega)\}
\end{displaymath}
is cofinal in $\alpha_0^+$ and hence has size $\alpha_0^+ > \alpha_0 \geq \kappa$. Applying Lemma \ref{above1}, for each $\beta' \in C$ we can pick $j(\beta')$ and $k(\beta') \neq k'(\beta')$ such that
\begin{align*}
(\alpha'_0, \beta', k(\beta')) \in a_{\alpha_0', \beta', j(\beta')}, &&
(\alpha'_0, \beta', k'(\beta')) \in a_{\alpha_0', \beta', j(\beta')}
\end{align*}
and therefore $|a_{\alpha'_0, \beta', j(\beta')}|>1$. We thus obtained a subset of $\brks{x}_{{\Sigma}_{\CM}^{\kappa}}$ of cardinality strictly greater than $\kappa$ whose elements are non-singleton sets. However, $\brks{x}_{{\Sigma}_{\kappa}}$ can have at most $\kappa$ non-singleton elements, for $|X|=\kappa$ and application of functions from ${\Sigma}_{\kappa}$ can only introduce either singletons or the empty set. Contradiction.
\end{proof}
\noindent We now go on to show that the $M_\delta$ disprove existence
of the tensor $\CM\tensor\PSet_\omega$.

\begin{lem}\label{subprop1}
Let $X$ be a set, $(x_i\mid i \in \omega)$ a sequence of elements of $X$, and $(a_i\mid i \in \omega)$ and $(b_i\mid i \in \omega)$ sequences of finite subsets of $X$. Then $(a_i \cup b_i\mid i \in \omega)$ subsumes $(x_i\mid i \in \omega)$ if and only if at least one of $(a_i\mid i \in \omega)$ and $(b_i\mid i \in \omega)$ does.
\end{lem}
\begin{proof}
The `if' direction is immediate. For the other direction, we employ Ramsey's theorem, which states that for any coloring of the set of pairs of natural numbers with two colours there is an infinite set of natural numbers which is monochromatic in the sense that all the pairs of numbers from that set are the same colour. Without loss of generality, the set $M$ witnessing that $(a_i \cup b_i \mid i \in \omega)$ subsumes $(x_i \mid i \in \omega)$ is the whole of $\omega$. Now for any $i < j < \omega$, we colour the pair $\{i,j\}$ red if $x_i \in a_j$ and blue otherwise. In this way we colour all 2-element subsets of $\omega$. By Ramsey's theorem, we can find a subset $M'$ of $\omega$ such that all pairs from $M'$ were given the same colour. If they were all coloured blue, then for each $i$ and $j$ in $M'$ with $i < j$ we have $x_i \in (a_j \cup b_j) \setminus a_j \subseteq b_j$, so $(b_i \mid i \in \omega)$ subsumes $(x_i \mid i \in \omega)$. A similar argument shows that if all pairs from $M'$ are coloured red then $(a_i \mid i \in \omega)$ subsumes $(x_i \mid i \in \omega)$. 
\end{proof}

\begin{lem}\label{commutes1}
  In $M_\delta$, each map $f_{\alpha, \beta, i}$ and the constant $c$
  commute with the bounded semilattice structure (i.e.\ the
  $\PSet_\omega$-algebra structure on $M_\delta$); in other words, any
  $M_{\delta}$ is a small $(\CM,\PSet_{\omega})$-tensor algebra.
\end{lem}
\begin{proof}
For $f_{\alpha,\beta,i}$, the case $\alpha\ge\delta$ is trivial; so assume $\alpha<\delta$. Since $(\emptyset \mid i \in \omega)$ never subsumes anything, we have $f_{\alpha, \beta, i}(\emptyset \mid (\alpha', \beta', i') \in B_{\alpha}) = \emptyset$, so $f_{\alpha, \beta, i}$ commutes with $\bot$. Because of the form of $f_{\alpha, \beta, i}$, to check that it commutes with $\vee$ it is enough to check that for any two families $(a_{\alpha', \beta', i'} \mid (\alpha', \beta', i') \in B_{\alpha})$ and $(b_{\alpha', \beta', i'} \mid (\alpha', \beta', i') \in B_{\alpha})$, the following two conditions are equivalent

\begin{iteMize}{$\bullet$}
\itum $\exists_{\cofinal} \alpha' < \alpha.~\exists_{\cofinal} \beta' < \alpha'^+.~ (a_{\alpha', \beta', i'} \cup b_{\alpha', \beta', i} \mid i' \in \omega) \text{ subsumes } ((\alpha', \beta', i')\mid i' \in \omega)$
\itum 
\begin{tabular}[t]{@{}r@{~~}l}
either &$\exists_{\cofinal} \alpha' < \alpha.~\exists_{\cofinal} \beta' < \alpha'^+.~(a_{\alpha', \beta', i'} \mid i' \in \omega) \text{ subsumes } ((\alpha', \beta', i')\mid i' \in \omega)$\\
 or &$\exists_{\cofinal} \alpha' < \alpha.~\exists_{\cofinal} \beta' < \alpha'^+.~(b_{\alpha', \beta', i'} \mid i' \in \omega) \text{ subsumes } ((\alpha', \beta', i')\mid i' \in \omega)$.
\end{tabular}
\end{iteMize}
This equivalence is immediate from Lemma \ref{subprop1} and the fact that a union of two subsets of an ordinal $\alpha$ is cofinal in $\alpha$ if and only if at least one of those two subsets is cofinal in $\alpha$. 

Finally, $c=\bot=\emptyset$ and hence $c\lor c= c$, i.e.\ $c$ also commutes with the semilattice structure.
\end{proof}

\begin{lem}\label{gen2}
  Any $M_{\kappa}$ is $\aleph_0$-reachable, specifically generated from
  the countable set $X = \{\{b\}\mid b \in B_1\}$ under the operations
  of\/ $\Sigma$ and the bounded semilattice structure.
\end{lem}
\begin{proof}
Let $A'$ be the subset of $M_{\kappa}$ generated in this way. It suffices to prove that for each $\alpha\leq\kappa$ we have $\{\{(\alpha, \beta, i)\} \mid \beta < \alpha^+ \wedge i < \omega\} \subseteq A'$. If $\alpha = 0$ this is true by definition. Otherwise, this is true by induction on $\alpha$, using the equation
$$\{(\alpha, \beta, i)\} = f_{\alpha, \beta, i}(\{(\alpha', \beta', j) \mid j < i'\}\mid(\alpha', \beta', i') \in B_{\alpha}).$$
We have now shown that $B \subseteq A'$, and $A \subseteq A'$ is immediate.
\end{proof}
\noindent By~\ref{cor:tensor-initial}, we obtain
\begin{thm}\label{thm:neg_main2}
  The tensor of $\CM$ and $\PSet_{\omega}$ does not exist (not even as
  a spurious tensor).
\end{thm}

\begin{rem}
  By Proposition~\ref{prop:tensor-inheritance}, the above result
  implies that every monad that is induced by a theory that has one
  binary operation and at most one constant and whose equations are
  implied by those of $\PSet_\omega$ (associativity, commutativity,
  and idempotence of the binary operation, neutrality of the constant
  if any) fails to be tensorable. In particular, finite non-empty
  powerset and both the full and the non-empty versions of the list
  monad and and the finite multiset monad, respectively, fail to be
  tensorable.
\end{rem}

\section{A Well-Order Monad That Fails to Be Tensorable}\label{sec:non-ex}
\noindent 

\noindent As announced above, we now present a second example of two
monads whose tensor fails to exist, originally published
in~\cite{GoncharovSchroder11}. It involves a well-order monad $\CW$,
where $\CW(X)$ consists of all well-orderings on non-empty subsets of
$X$, plus an error element; the other partner is ranked, a free monad
over two binary operations.  The construction follows the same pattern
as in the preceding section: first we introduce a large 
theory of $\CT_{\CW}$-algebras, then we show that this theory yields a
monad $\CW$, and finally we prove that the tensor product $\CW$ with
the ranked partner does not exist using
Corollary~\ref{cor:tensor-initial}.
\begin{defi}\label{def:wo}
  The theory $\CT_\CW$ of \emph{strict non-empty well-orders} has a
  signature consisting of a constant $\bot$ and a family of operation
  symbols $\ulist_{\kappa}$ of arity $\kappa$, indexed over all
  positive ordinals $\kappa$. It imposes the following equations.

%
  \begin{enumerate}[(1)]
  \item \emph{Strictness:} $\ulist_{\kappa}(w_{\alpha}\mid \alpha <
    \kappa) = \bot$ whenever $w_\alpha=\bot$ for some $\alpha<\kappa$.
  \item \emph{Non-repetitiveness:} $\ulist_{\kappa}(w_{\alpha} \mid
    \alpha < \kappa) = \bot$ whenever $w_{\alpha_1}=w_{\alpha_2}$
    for some $\alpha_1<\alpha_2<\kappa$.
  \item\label{item:assoc} \emph{Associativity:} For every
    small-ordinal-indexed family $(\kappa_\mu)_{\mu<\nu}$ of ordinals
    $\kappa_\mu>0$,
    \begin{displaymath}
      \ulist_{\kappa}(w_{\mu,\alpha} \mid \mu < \nu, \alpha < \kappa_{\mu})=\ulist_{\nu}(\ulist_{\kappa_{\mu}}(w_{\mu, \alpha} \mid \alpha < \kappa_{\mu}) \mid \mu < \nu)
    \end{displaymath}
    where on the left hand side $\kappa=\sum_{\mu<\nu}\kappa_{\mu}$ is
    regarded as having elements $\brks{\mu,\alpha}$ with $\mu<\nu$ and
    ${\alpha<\kappa_\mu}$.
  \end{enumerate}
\end{defi}
\noindent We regard an ordinal $\kappa$ as the set of all ordinals
$\alpha<\kappa$ unless we explicitly specify otherwise, as in the
associativity law above where we use a more convenient isomorphic
representation of ordinal sums.  Even though in the above formulations
of strictness and non-repetitiveness we employ the word `whenever',
they may nevertheless be interpreted as sets of equational axioms.

Now consider a small $\CT_{\CW}$-algebra $X$. By non-repetitiveness,
for every $\kappa$ whose cardinality exceeds $|X|$, $\ulist_{\kappa}$
is identically $\bot$, which means that the set of nontrivial
operations in the structure of any particular $\CT_{\CW}$ is small. A
homomorphism of two $\CT_{\CW}$-algebras $(X,\ulist_{\kappa})$ and
$(Y,\ulist_{\kappa})$ is a map $f:X\to Y$ that commutes with the
operations, i.e.\
\begin{equation*}
f(\ulist_{\kappa}(w_{\alpha}\mid \alpha < \kappa)) = \ulist_{\kappa}(f(w_{\alpha}) \mid \alpha < \kappa)\quad\text{for $w\in X^\kappa$.}
\end{equation*} 

\begin{lemdefn}
  The theory $\CT_{\CW}$ has small free algebras. The induced monad
  $\CW$, the \emph{strict non-empty well-order monad}, maps a small
  set $X$ to the set
  \begin{equation*}
    \CW X=\{(Y,\rho)\mid \emptyset\neq Y\subseteq X,\rho\text{ a well-order on Y}\}\cup\{\bot\}.
\end{equation*}
Its unit maps $x\in X$ to unique well-order on $\{x\}$, and its
multiplication concatenates well-orders in case all its arguments are
well-orders whose carriers are pairwise disjoint (so that the result
is again a well-order), and otherwise returns $\bot$.
\end{lemdefn}
\noindent One may alternatively think of the strict non-empty
well-order monad as a monad of infinite non-repetitive non-empty
lists, with $\bot$ playing the role of an error element that is thrown
in case of repetitions arising by concatenation, and that is
propagated through concatenation by the strictness law.
\begin{proof}
  It is easy to see that the elements of $\CW(X)$ serve as unique
  normal forms in $F_{\CT_\CW}(X)$. 
\end{proof}
  

%

%
\noindent The second monad for our example is very simple, and has
finite rank: Let $\Sigma_{2,2}^{\varstar}$ be the free algebra monad for the empty theory in the signature $\Sigma_{2,2}$ consisting of just 2 binary operations.

\begin{lem}\label{lem:kappa-chain}
  For every infinite cardinal $\kappa$, there exists a $2$-reachable small
  $(\CW,\Sigma_{2,2}^{\varstar})$-tensor algebra $W_{\kappa}$ such
  that $|W_\kappa|>\kappa$.
\end{lem}
\begin{proof}
  The domain of $W_{\kappa}$ is the union $\{\bot,0,1\}\cup
  U_{\kappa}^0\cup U_{\kappa}^1\cup L_{\kappa}$ where the
  $U_{\kappa}^i$ and $L_{\kappa}$ are sets of terms defined by infinitary mutual
  recursion according to the the rules
\begin{align*}
\infer
{\brks{i,0,t}\in U_{\kappa}^i}
{t\in W_{\kappa}-\{0\}}
\end{align*}
where $i\in\{0,1\}$, and
\begin{align*}
\infer
{t\in L_{\kappa}}
{
t:\nu\inject U_{\kappa}^0 \cup U_{\kappa}^1 ~~~~~~
\forall\mu.~\mu+1<\nu\implies\bigl(t(\mu)\in U_{\kappa}^{0}\iff t(\mu+1)\in U_{\kappa}^1\bigr)
}
\end{align*}
where $\nu$ is an ordinal such that $1<|\nu|\leq\kappa$ and $\into$ is
read as $t$ being injective (not a subset inclusion). Notice that
$U^0_\kappa\cap U^1_\kappa=\emptyset$, so the second premise says that
$t(\mu)$ alternates between $U^0_\kappa$ and $U^1_\kappa$. Let us
define a length map $\len{}$ from $W_{\kappa}$ to ordinals as follows:
we put $\len{t}=1$ for $t\in\{\bot,0,1\} \cup U_{\kappa}^0\cup
U_{\kappa}^1$, and $\len{t}=\nu$ whenever $t:\nu\inject U_{\kappa}^0
\cup U_{\kappa}^1\in L_{\kappa}$. Note that this implies $\len{t}>1$
iff $t\in L_{\kappa}$.

To give a $\Sigma^{\varstar}_{2,2}$-algebra structure over $W_{\kappa}$ is the same as to define two binary maps $u_0,u_1:W_{\kappa}\times W_{\kappa}\to W_{\kappa}$. For $i=0,1$ we put by definition
\begin{iteMize}{$\bullet$}
\item $u_i(t,t)=t$ if $t\in\{0,1\}$ ;
\item $u_i(0,t)=\brks{i,0,t}\in U_{\kappa}^i$ whenever $t\in\{1\}\cup L_{\kappa}$;
\item $u_i(s,t)=\bot$ in the remaining cases.
\end{iteMize}
We now define a $\CT_\CW$-algebra structure on $W_\kappa$. We
interpret $\bot$ by $\bot$, and $\ulist_1$ by $\id$. For $\nu>1$ and
$t\in (W_{\kappa})^{\nu}$ we define $\ulist_{\nu}(t)$ by the clauses
\begin{iteMize}{$\bullet$}
\item $\ulist_{\nu}(t)=s$, provided the map $s:\zeta\to W_\kappa$ on
  $\zeta=\sum_{\mu<\nu}\len{t(\mu)}$ defined as follows is in
  $L_{\kappa}$: We regard $\zeta$ as consisting of pairs
  $\brks{\mu,\kappa}$ where $\mu<\nu$ and $\kappa<\len{t(\mu)}$. For
  every such $\brks{\mu,\kappa}$, put $s\brks{\mu,\kappa} =
  t(\mu)(\kappa)$ if $t(\mu)\in L_{\kappa}$, and $s\brks{\mu,\kappa} =
  t(\mu)$ otherwise (in which case necessarily $\kappa=0$).
  \item $\ulist_{\nu}(t)=\bot$ otherwise.
\end{iteMize}
It is then clear by construction that $W_\kappa$ is $2$-reachable (it
is generated by $0$ and $1$), as the rules defining $L_\kappa$ and the
$U^i_\kappa$ just amount to closure under the $u_i$ and $\ulist_\nu$
as defined above. Next, we have to check that $W_\kappa$ is really a
$\CT_\CW$-algebra. By definition, for every $t$, $\ulist_{\nu}(t)\in
L_{\kappa}\cup\{\bot\}$, hence the conditions~(1) and~(2) of
Definition~\ref{def:wo} are ensured automatically. Condition~(3) is
less trivial, but still routine. Finally we need to verify the tensor
law. In the case at hand it amounts to proving the equation
\begin{displaymath}
u_i(\ulist_{\nu}(t),\ulist_{\nu}(s)) = \ulist_{\nu}(\lambda\mu<\nu.\ u_i(t(\mu),s(\mu)))
\end{displaymath}
for every $s,t\in W_{\kappa}$, $i=0,1$. It is immediate by definition that both sides of this equation equal $\bot$ unless $\nu=1$. In the latter case the equation also follows since, by definition, $\ulist_1=\id$.

Finally, we show that $|W_{\kappa}|>\kappa$. In order to derive a contradiction, assume that $|W_{\kappa}|\le\kappa$ and let $\varsigma$ be an ordinal number such that $|W_{\kappa}|=|\varsigma|$. Let $\rho$ be a bijection $\varsigma\to W_{\kappa}-\{0\}$. Since $|\varsigma|\le\kappa$ and hence $|\varsigma\cdot 2|\le\kappa$ (since $\kappa$ is infinite), we can form an element $t_{\rho}:\varsigma\cdot 2\into U_{\kappa}^0\cup U_{\kappa}^1$ of $W_{\kappa}$ by putting $t_{\rho}(\varsigma',i)=\brks{i,0,\rho(\varsigma')}$ for $\varsigma'<\varsigma$, $i=0,1$. By varying $\rho$, we can produce as many such elements as there are isomorphisms from $\varsigma$ to $W_{\kappa}-\{0\}$, i.e.\ strictly more than $|\varsigma|=|W_{\kappa}|$, contradiction. 
\end{proof}

\noindent By Corollary~\ref{cor:tensor-initial} we obtain
\begin{thm}\label{thm:neg_main}
  The tensor of the strict non-empty well-order monad $\CW$ and
  $\Sigma_{2,2}^{\varstar}$ does not exist, even as a spurious tensor.
\end{thm}

\section{Conclusion}

\noindent Tensors of theories, or monads, capture algebras of
one theory in the category of algebras of the other. For unranked
monads, equivalently large theories with small free
algebras, existence of tensors is not self-understood; we call a
theory or monad \emph{tensorable} if its tensors with all other
theories, or monads, respectively, exist. We have given two
counterexamples to tensorability of monads:
\begin{iteMize}{$\bullet$}
\item the tensor of the finite powerset monad with a certain somewhat
  complex unranked monad fails to exist;
\item the tensor of the strict nonempty well-order monad and a
simple finitary monad, generated by two binary operations and no
equations, fails to exist.
\end{iteMize}
We have thus settled in the negative the long-standing open question
of universal existence of tensors of monads on $\Set$~\cite{Manes69},
which has recently reemerged in the perspective of work on algebraic
effects~\cite{HylandPlotkinEtAl06,HylandLevyEtAl07}. The negative
answer as such is in accordance with expectations, but the actual
counterexamples are rather different from what was previously
suspected. 

In addition to our negative results, we have established a positive
result stating that all bounded powerset monads---except finite
powerset---are (genuinely) tensorable.

Our main motivation for the study of tensors as such is to develop a
monadic framework for non-interference of side-effects, noting that
the tensor law precisely amounts to orthogonality of the component
monads; these ideas will be further developed in future
research. Another topic of further interest is the investigation of
tensors over base categories other than $\Set$, for example the
category of $\omega$-complete partial orders.

\subsubsection*{Acknowledgements} We wish to thank various
contributors to the categories mailing list, in particular Peter
Johnstone, for useful insights communicated via the list, and the
anonymous referees of~\cite{GoncharovSchroder11} for valuable
pointers to the literature.

\bibliographystyle{myabbrv}
\bibliography{monads}

\begin{thebibliography}{10}

\bibitem{AdamekBowlerLevyMilius:coprodmonset}
J.~Ad{\'a}mek, N.~Bowler, P.~B. Levy, and S.~Milius.
\newblock Coproducts of monads on set.
\newblock In N.~Dershowitz, ed., {\em Logic in Computer Science, LICS 2012},
  pp. 45--54. IEEE, 2012.

\bibitem{BarrWells85}
M.~Barr and C.~Wells.
\newblock {\em Toposes, Triples and Theories}, vol. 278 of {\em Grundlehren der
  mathematischen Wissenschaften}.
\newblock Springer, 1985.

\bibitem{CenciarelliMoggi93}
P.~Cenciarelli and E.~Moggi.
\newblock A syntactic approach to modularity in denotational semantics.
\newblock In {\em Category Theory and Computer Science, CTCS 1993}, 1993.

\bibitem{DrosteKuichEtAl09}
M.~Droste, W.~Kuich, and H.~Vogler, eds.
\newblock {\em Handbook of Weighted Automata}.
\newblock Springer, 2009.

\bibitem{Dubuc70}
E.~Dubuc.
\newblock {\em Kan Extensions in Enriched Category Theory}, vol. 145.
\newblock Springer, 1970.

\bibitem{Felgner:NBGconserve}
U.~Felgner.
\newblock Comparison of the axioms of local and universal choice.
\newblock {\em Fundamenta Mathematicae}, 71:43--62, 1971.

\bibitem{Freyd66}
P.~Freyd.
\newblock The theory of functors and models.
\newblock In {\em Theory of Models --- Proceedings of the 1963 International
  Symposium at Berkeley}, pp. 107--120. North Holland, 1966.

\bibitem{GoncharovSchroder11}
S.~Goncharov and L.~Schr{\"o}der.
\newblock A counterexample to tensorability of effects.
\newblock In A.~Corradini and B.~Klin, eds., {\em Algebra and Coalgebra in
  Computer Science, CALCO 2011}, Lect.\ Notes Comput.\ Sci., pp. 208--211.
  Springer, 2011.

\bibitem{GoncharovSchroder11b}
S.~Goncharov and L.~Schr{\"o}der.
\newblock Powermonads and tensors of unranked effects.
\newblock In M.~Grohe, ed., {\em Logic in Computer Science, LICS 2011}, pp.
  227--236. IEEE Computer Society, 2011.

\bibitem{GoncharovSchroderEtAl09}
S.~Goncharov, L.~Schr\"{o}der, and T.~Mossakowski.
\newblock Kleene monads: handling iteration in a framework of generic effects.
\newblock In A.~Kurz and A.~Tarlecki, eds., {\em Algebra and Coalgebra in
  Computer Science, CALCO 2009}, vol. 5728 of {\em Lect.\ Notes Comput.\ Sci.},
  pp. 18--33. Springer, 2009.

\bibitem{HylandLevyEtAl07}
M.~Hyland, P.~B. Levy, G.~Plotkin, and J.~Power.
\newblock Combining algebraic effects with continuations.
\newblock {\em Theoret.\ Comput.\ Sci.}, 375(1-3):20 -- 40, 2007.
\newblock Festschrift for John C. Reynolds's 70th birthday.

\bibitem{HylandPlotkinEtAl06}
M.~Hyland, G.~Plotkin, and J.~Power.
\newblock Combining effects: Sum and tensor.
\newblock {\em Theoret.\ Comput.\ Sci.}, 357:70--99, 2006.

\bibitem{HylandPower07}
M.~Hyland and J.~Power.
\newblock The category theoretic understanding of universal algebra: Lawvere
  theories and monads.
\newblock In {\em Computation, Meaning, and Logic: Articles dedicated to Gordon
  Plotkin}, vol. 172 of {\em Electron.\ Notes Theoret.\ Comput.\ Sci.}, pp.
  437--458. Elsevier, 2007.

\bibitem{Kelly:transfin}
G.~M. Kelly.
\newblock A unified treatment of transfinite constructions for free algebras,
  free monoids, colimits, associated sheaves, and so on.
\newblock {\em Bull. Austral. Math. Soc.}, 22:1--84, 1980.

\bibitem{Linton66}
F.~Linton.
\newblock Some aspects of equational categories.
\newblock In {\em Conference on Categorical Algebra, La Jolla}, pp. 84--94.
  Springer, 1966.

\bibitem{LuthGhani02}
C.~L\"{u}th and N.~Ghani.
\newblock Composing monads using coproducts.
\newblock In {\em International Conference on Functional Programming, ICFP
  2002}, vol. 37(9) of {\em SIGPLAN Not.}, pp. 133--144. ACM, 2002.

\bibitem{Manes69}
E.~Manes.
\newblock A triple theoretic construction of compact algebras.
\newblock In {\em Seminar on Triples and Categorical Homology Theory}, vol.~80
  of {\em Lect.\ Notes Math.}, pp. 91--118. Springer, 1969.

\bibitem{Moggi91}
E.~Moggi.
\newblock Notions of computation and monads.
\newblock {\em Inf.\ Comput.}, 93:55--92, 1991.

\bibitem{Moggi95}
E.~Moggi.
\newblock A semantics for evaluation logic.
\newblock {\em Fund.\ Inform.}, 22:117--152, 1995.

\bibitem{Morse:sets}
A.~P. Morse.
\newblock {\em A theory of sets}.
\newblock Academic Press, 1965.

\bibitem{Peyton-Jones03}
S.~Peyton-Jones, ed.
\newblock {\em {Haskell} 98 Language and Libraries --- The Revised Report}.
\newblock Cambridge University Press, 2003.
\newblock Also: J.\ Funct.\ Prog.\ {{\bf 13}} (2003).

\bibitem{PlotkinPower02}
G.~Plotkin and J.~Power.
\newblock Notions of computation determine monads.
\newblock In M.~Nielsen and U.~Engberg, eds., {\em Foundations of Software
  Science and Computation Structures, FOSSACS 2002}, vol. 2303 of {\em Lect.\
  Notes Comput.\ Sci.}, pp. 342--356. Springer, 2002.

\bibitem{PowerShkaravska04}
J.~Power and O.~Shkaravska.
\newblock From comodels to coalgebras: State and arrays.
\newblock In J.~Ad{\'a}mek and S.~Milius, eds., {\em Coalgebraic Methods in
  Computer Science, CMCS 2004}, vol. 106 of {\em Electron.\ Notes Theoret.\
  Comput.\ Sci.}, pp. 297--314, 2004.

\bibitem{SchroderMossakowski04}
L.~Schr{\"o}der and T.~Mossakowski.
\newblock Generic exception handling and the {J}ava monad.
\newblock In C.~Rattray, S.~Maharaj, and C.~Shankland, eds., {\em Algebraic
  Methodology and Software Technology, AMAST 2004}, vol. 3116 of {\em Lect.\
  Notes Comput.\ Sci.}, pp. 443--459. Springer, 2004.

\bibitem{Scott55}
D.~Scott.
\newblock Definitions by abstraction in axiomatic set theory.
\newblock {\em Bull.\ AMS}, 61(442):8, 1955.

\bibitem{Wadler97}
P.~Wadler.
\newblock How to declare an imperative.
\newblock {\em ACM Comput.\ Surveys}, 29:240--263, 1997.

\end{thebibliography}

\end{document}
